\documentclass[letterpaper, 10 pt, conference]{ieeeconf}  

\IEEEoverridecommandlockouts                              

\usepackage{cite}
\usepackage{amsmath,amssymb,amsfonts}
\usepackage{graphicx}
\usepackage{textcomp}

\usepackage[utf8]{inputenc}
\usepackage{graphics} 
\usepackage{epsfig} 
\usepackage{times} 

\usepackage{graphicx}      
\usepackage{siunitx}
\usepackage{pgfplots}
\usepackage{tikz,amsmath}
\usetikzlibrary{external}
\usepackage{tikzscale}
\usepackage{algpseudocode}
\usepackage{amsthm}
\usepackage{algorithm}
\usepackage{fancyhdr}
\usepackage{subcaption} 
\usepackage{eso-pic}
\usepackage{mathdots}
\usepackage{accents}

\usetikzlibrary{shapes,arrows}
\usetikzlibrary{positioning}

\let\labelindent\relax
\usepackage{enumitem}

\newtheorem{thm}{Theorem}
\newtheorem{lem}[thm]{Lemma}
\newtheorem{defn}[thm]{Definition}

\theoremstyle{plain}

\pgfplotsset{compat=newest}
\pgfplotsset{plot coordinates/math parser=false}
\newlength\figureheight
\newlength\figurewidth

\newlength\defcolwidth

\newcommand*{\tran}{^{\mkern-1.5mu\mathsf{T}}}

\usepackage{tikz}
\usepackage{tikzscale}

\definecolor{myorange}{cmyk}{0,0.35,0.85,0} 
\definecolor{mypurple}{cmyk}{0.5,1,0,0} 

\definecolor{matblue1}{rgb}{0,0.4470,0.7410}
\definecolor{matred1}{rgb}{0.85,0.325,0.098}
\definecolor{matyel1}{rgb}{0.9290, 0.6940, 0.1250}
\definecolor{matpur1}{rgb}{0.4940, 0.1840, 0.5560}
\definecolor{matgre1}{rgb}{0.4660, 0.6740, 0.1880}
\definecolor{matblue2}{rgb}{0.3010, 0.7450, 0.9330}
\definecolor{matred2}{rgb}{0.6350, 0.0780, 0.1840}
\definecolor{matgrey1}{rgb}{0.5, 0.6, 0.7}
\definecolor{matpink1}{rgb}{1, 0.07, 0.65}
\definecolor{matblue3}{rgb}{0.07, 0.62, 1}

\newcommand{\greendashdot}{\raisebox{2pt}{\tikz{\draw[-,matgre1,densely dash dot,line width = 0.9pt](0,0) -- (3mm,0);}}}

\newcommand{\purdots}{\raisebox{2pt}{\tikz{\draw[-,matpur1,densely dotted,line width = 0.9pt](0,0) -- (3mm,0);}}}

\newcommand{\bluedash}{\raisebox{2pt}{\tikz{\draw[-,matblue1,dashed,line width = 0.9pt](0,0) -- (3mm,0);}}}

\newcommand{\redline}{\raisebox{2pt}{\tikz{\draw[-,matred1,solid,line width = 0.9pt](0,0) -- (3mm,0);}}}

\newcommand{\yelline}{\raisebox{2pt}{\tikz{\draw[-,matyel1,solid,line width = 0.9pt](0,0) -- (3mm,0);}}}

\newcommand*\rev[1]{{\color{black}#1}}

\title{\LARGE \bf
Conjugate Gradient MIMO Iterative Learning Control \\ Using Data-Driven Stochastic Gradients*
}%

\author{Leontine Aarnoudse$^{1}$ and Tom Oomen$^{1,2}$
	\thanks{*This work is part of the research programme VIDI with project number 15698, which is (partly) financed by the NWO.}
	\thanks{$^{1}$Leontine Aarnoudse and Tom Oomen are with the Dept. of Mechanical Engineering, Control Systems Technology, Eindhoven University of Technology, Eindhoven, The Netherlands.}%
	\thanks{$^{2}$Tom Oomen is also with the Delft Center for Systems and Control, Delft University of Technology, Delft, The Netherlands.
		{\tt\small l.i.m.aarnoudse@tue.nl, t.a.e.oomen@tue.nl}}%
}%

\begin{document}
\AddToShipoutPictureBG*{%
	\AtPageUpperLeft{%
		\setlength\unitlength{1in}%
		\hspace*{\dimexpr0.5\paperwidth\relax}
		\makebox(0,-1)[c]{
			\parbox{\paperwidth}{ \centering
				Leontine Aarnoudse and Tom Oomen, Conjugate Gradient MIMO Iterative Learning Control Using Data-Driven Stochastic Gradients, \\
				In {\em 60th IEEE Conference on Decision and Control}, Austin, Texas, USA, 2021}}%
}}

\maketitle%
\thispagestyle{empty}%
\pagestyle{empty}%

\setlength\defcolwidth{7.85cm}

\setlength\figurewidth{.9\defcolwidth}
\setlength\figureheight{.7\figurewidth}

\begin{abstract}
Data-driven iterative learning control can achieve high performance for systems performing repeating tasks without the need for modeling. The aim of this paper is to develop a fast data-driven method for iterative learning control that is suitable for massive MIMO systems through the use of efficient unbiased gradient estimates. A stochastic conjugate gradient descent algorithm is developed that uses dedicated experiments to determine the conjugate search direction and optimal step size at each iteration. The approach is illustrated on a multivariable example, and it is shown that the method is superior \rev{to both the earlier stochastic gradient descent and deterministic conjugate gradient descent methods.}   

\end{abstract}

\section{Introduction} \label{sec:intro}

Direct data-driven approaches are advantageous in many control problems, because they avoid the costly process of modeling and identification \cite{Gevers2002,Hjalmarsson2005,Dorfler2021}, and do not suffer from performance limitations due to model uncertainties. Examples of data-driven methods include procedures for identifying system norms \cite{Hjalmarsson2001,Oomen2014} and tuning feedback controllers \cite{Hjalmarsson1999,Campi2002,Vanwaarde2020}.

In iterative learning control (ILC), measured data is used in conjunction with approximate models in order to design feedforward signals that are capable of rejecting repeating disturbances completely. \rev{This significantly increases the performance of systems that perform repeating tasks.} Examples of ILC frameworks include frequency-domain based inverse model ILC \cite{VanZundert2018c,Bristow2006a}, Arimoto-type ILC \cite{Arimoto1984}, and optimization-based approaches \rev{such as} norm-optimal ILC \cite{Gunnarsson2001,Owens2016} and gradient- and coordinate-descent \rev{ILC} \cite{Wu2018,Chen2019}. These partially model-based approaches require system knowledge in the form of invertible models in frequency-domain ILC, or certain properties of the system's Markov parameters. Since models are always approximate, the methods provide robustness against model uncertainty through $Q$-filters in frequency-domain ILC, regularization in norm-optimal ILC, or robust design \cite{VandeWijdeven2009,Owens2014,Son2016}.


In \cite{Bolder2018}, a model-free adjoint ILC algorithm is introduced, in which experiments on the adjoint system \cite{Wahlberg2010} are used to obtain the gradient of a cost criterion. These gradients are used in a gradient-descent type ILC algorithm \cite{Furuta1987}, enabling a complete data-driven design for MIMO systems. The approach is related to data-driven ILC approaches such as \rev{extremum-seeking based ILC} \cite{Khong2016}. \rev{In \cite{Bolder2018}, as well as in} comparable approaches for MIMO experiment-based iterative feedback tuning \cite{Hjalmarsson1999} and $H_\infty$-norm estimation \cite{Oomen2014}, the gradient of an $n_i \times n_o$ MIMO system is generated through $n_i \times n_o$ experiments. Because of the high number of experiments per iteration required in model-free adjoint ILC, the method does not scale well for massive MIMO systems.

Model-free ILC for massive MIMO systems is further developed in \cite{Aarnoudse2020c}, where the experimentally expensive deterministic gradient from \cite{Bolder2018} is replaced by an unbiased \rev{gradient} estimate, obtained from a single experiment. The estimate is used in a stochastic approximation adjoint ILC (SAAILC) algorithm based on stochastic gradient descent. \rev{The SAAILC algorithm requires far fewer experiments to reach the same cost compared to the deterministic gradient descent ILC in \cite{Bolder2018}.} However, in terms of the number of iterations needed to converge, gradient descent algorithms are known to be slow due to their lack of curvature information. In \cite{Bolder2018}, first steps are made towards a data-driven quasi-Newton adjoint ILC algorithm using \rev{Broyden-Fletcher-Goldfarb-Shanno (BFGS)} updates \rev{to increase the convergence speed}, but this approach is not directly applicable to stochastic gradients.


Although several steps have been made towards efficient data-driven iterative learning control, an \rev{optimal data-driven approximate approach} is underdeveloped. The aim of this paper is to develop an approach which converges faster than standard gradient descent, while using experimentally efficient gradient estimates. The contribution of this paper is threefold:
\begin{itemize}
	\item A conjugate gradient ILC algorithm using unbiased gradient approximations is developed. It is shown that the standard expressions for conjugate gradient methods are not applicable in case of stochastic gradients. Instead, two additional experiments are used to determine the conjugate search direction and the optimal step size.
	\item An analysis of related methods, \rev{including gradient descent ILC and deterministic conjugate gradient descent ILC} is provided. 
	\item The proposed approach is illustrated using a random MIMO system.
\end{itemize}

Preliminary research results related to improving gradient estimation, in particular by obtaining unbiased estimates in a fast manner, are presented in \cite{Aarnoudse2020c}. \rev{There, the} gradient estimates were used in a relatively naive gradient descent algorithm. The current paper contains the above three contributions that are completely new.

The problem of increasing convergence speed for stochastic optimization algorithms is well-studied in the field of machine learning, where the use of mini-batching results in non-deterministic algorithms. As classical quasi-Newton methods are not applicable in stochastic situations \cite{Bottou2018}, much research is aimed at developing stochastic quasi-Newton methods. Examples include the work by \cite{Schraudolph2007} on stochastic BFGS and limited-memory BFGS. In \cite{Mokhtari2014}, regularization is added to the stochastic BFGS method, and \cite{Moritz2016} combines stochastic BFGS methods with variance reduced stochastic gradients. A recent overview is given by \cite{Mokhtari2020}. These methods typically use multiple gradient evaluations of one mini-batch to obtain locally deterministic Hessian estimates, a consistency assumption that cannot be satisfied for the stochastic gradients used in ILC, see \cite{Aarnoudse2020c}. 

Recently, methods that do not depend on this consistency have been developed in \cite{Wills2021}, where Gaussian processes are used to model the inverse Hessian, and \cite{Wills2020}, where the quasi-Newton secant condition is relaxed in a direct least-squares estimate of the inverse Hessian. However, model-free ILC uses quadratic objectives, for which the Hessian is constant and conjugate gradients are the method of choice \cite[Section 8.3]{Murphy}. Therefore, in this paper an approach is proposed that is based on conjugate gradients.

The paper is structured as follows. In Section \ref{sec:problem}, the problem considered in this paper is introduced. In Section \ref{sec:approach}, the proposed conjugate gradient ILC algorithm is developed. In Section \ref{sec:comp}, related methods are analyzed in comparison to the proposed method. In Section \ref{sec:sims} the approach is illustrated using simulations. Conclusions are given in Section \ref{sec:conclusions}.

%
%
%
%

\section{Problem formulation} \label{sec:problem}

In this section, the problem considered in this paper is introduced. Consider \rev{the aim of finding a control signal that minimizes the error of a system, which is expressed by the criterion}
\begin{align} \label{eq:cost}
	\mathcal{J}(f) = \|r-Jf\|^2_2.
\end{align}
Here, $\|x\|_2 = \sqrt{x\tran x}$ and the \rev{unknown} MIMO system $J$ with $n_i$ inputs and $n_o$ outputs is given in lifted form by 
\begin{align} 
\underbrace{\begin{bmatrix} e^1 \\ \vdots \\ e^{n_o} \end{bmatrix}}_e &= 	 \underbrace{\begin{bmatrix} r^1 \\ \vdots \\ r^{n_o} \end{bmatrix}}_r - \underbrace{\begin{bmatrix} J^{11} & \dots & J^{1n_i} \\ \vdots & & \vdots \\ J^{n_o 1} & \dots & J^{n_o n_i} \end{bmatrix}}_J \underbrace{\begin{bmatrix} f^1 \\ \vdots \\ f^{n_i} \end{bmatrix}}_f 
\end{align} 
with input $f$, error $e$, unknown exogenous disturbance $r$ \rev{and output $\rev{y = Jf}$}. Here, $J^{lm}\in \mathbb{R}^{N \times N}$ for finite signal length $N \in \mathbb{Z}^+$, and  $y^l, e^l, r^l, f^m \in \mathbb{R}^{N \times 1}$ for $l=1,...,n_o$, $m=1,...,n_i$. An example of system $J$, which can represent both open-loop and closed-loop systems, is shown in Figure \ref{fig:control_scheme}.

The aim of this paper is to develop an efficient data-driven approach to minimizing (\ref{eq:cost}). To this end, judiciously chosen experiments experiments are used to generate the gradient $\frac{\partial \mathcal{J}}{\partial f}$ (more specifically, an unbiased estimate thereof). The approach, which is based on conjugated gradients, is introduced in the next section.

\begin{figure}[t]	
	\centering
	\includegraphics{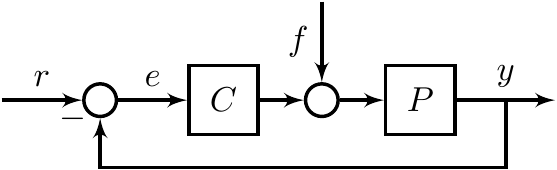}
	\caption{Closed-loop system with reference $r$, error $e$ and feedforward input $f$, for which $J = P(I+CP)^{-1}$ is the process sensitivity of the closed-loop system.}
	\label{fig:control_scheme}
\end{figure}

\section{Stochastic conjugate gradient ILC} \label{sec:approach}
In this section, model-free conjugate gradient ILC is introduced. First, a suitable search direction is found. Secondly, the optimal step size in this direction is computed. Thirdly, it is explained how unbiased gradient estimates can be obtained from experiments and lastly, the implementation of model-free conjugate gradient ILC is explained.

Consider again criterion (\ref{eq:cost}), and assume that estimates $\hat{g}(f)$ of the gradient $g(f) = \frac{\partial \mathcal{J}}{\partial f}$ can be obtained for which
\begin{align*}
	\mathbb{E}(\hat{g}(f)) = g(f),
\end{align*}
i.e., the estimates $\hat{g}(f)$ are unbiased. In addition, while the system $J$ is unknown, \rev{noise-free} evaluations of $Jf$ are available through experiments. Since the criterion $\mathcal{J}(f)$ is quadratic, it can be minimized by setting the gradient $g(f)$ equal to zero, i.e., by solving $g(f) = 0$ with
\begin{align} \label{eq:grad}
	g(f) = 2 J\tran J f- 2 J\tran r .
\end{align}
This is equivalent to solving a system of linear equations, given by
\begin{align} \label{eq:axb}
	2 J\tran J f = 2 J\tran r.  
\end{align} 
Since the system $J$ is unknown in a model-free setting, (\ref{eq:axb}) is solved iteratively based on data by using parameter updates of the form
\begin{align} \label{eq:fjplus1}
	f_{j+1} &= f_j + \varepsilon_j p_j,
\end{align} 
with step size $\varepsilon_j$ and search direction $p_j$. Note that for $p_j = g(f_j)$, a gradient descent algorithm is recovered. \rev{If a model $J$ is available, standard norm-optimal ILC \cite{Gunnarsson2001} is recovered by taking $\varepsilon_j p_j = (J\tran J)^{-1} J\tran e_j$.}

\subsection{Conjugate search directions}
The main idea of conjugate gradient descent for quadratic problems is that fast convergence can be achieved by minimizing \rev{the criterion (\ref{eq:cost})} along a sequence of conjugated gradient directions.

\begin{defn}
 Two vectors $x$ and $y$ are $A$-conjugate if
\begin{align}
	x\tran A y = 0. 
\end{align}	
\end{defn} 
	
Let $g_j$ denote the gradient at iteration $j$, i.e., $g_j = g(f_j)$. By taking the initial search direction $p_1$ equal to the initial gradient $g_1$ and choosing all subsequent search directions such that $i\neq j \implies p_i \tran J\tran J p_j = 0$ for $1 \leq i < j \leq N n_i$, i.e., $p_i$ and $p_j$ are conjugate with respect to $J\tran J$, a sequence of Krylov subspaces is generated \cite[Section 11.3]{Golub2013} that is given by 
\begin{align}
	S_j &= \mathcal{K}(J\tran J, g_1,j) \nonumber \\ 
	&= \text{span} \left\{ g_1, J\tran J g_1, (J\tran J)^2 g_1,..., (J\tran J)^{j-1} g_1 \right\}.
\end{align}  
For each iteration, an exact line search is used that ensures that $f_j$ solves
\begin{align} 
	\min_{f_j \in \tilde{S}_j} \mathcal{J}(f_j),
\end{align}
with
\begin{align}
	\tilde{S}_j = \text{span} \left\{f_1, g_1, J\tran J g_1,..., (J\tran J)^{j-1} g_1 \right\}.
\end{align}
Since $\tilde{S}_{j+1}$ includes both $f_j$ and the gradient $g_j$, it is guaranteed that the update based on conjugated gradients is at least as good as the steepest descent update \cite{Allwright1976}.

For conjugate gradient ILC, the initial search direction $p_1$ in (\ref{eq:fjplus1}) is chosen equal to the unbiased gradient estimate $\hat{g}_1 = \hat{g}(f_1)$. Subsequent search directions are given by
\begin{align} \label{eq:directions}
	p_{j+1} &= \hat{g}_{j+1} + \tau_j p_j, 
\end{align}
where the scalar $\tau_j$ is chosen such that the direction $p_{j+1}$ is $J\tran J$-conjugate to the previous search direction $p_{j}$. Since $\hat{g}_j \neq g_j$, standard conjugate gradient expressions cannot be applied, as is shown in Section \ref{sec:comp}. For the stochastic conjugate gradient approach using gradient estimate $\hat{g}_j$, the expression for $\tau_j$ is given in the following theorem.
\begin{thm} \label{thm:directions}
	The search directions $p_{j+1}$ and $p_j$ are $J\tran J$-conjugate if 
	\begin{align}  \label{eq:tau}
		\tau_j = - \frac{(J p_j)\tran (J  \hat{g}_{j+1})}{(J p_j)\tran (J p_j)}
	\end{align}
\end{thm}

\begin{proof}
	Because $p_{j+1}$ and $p_j$ are chosen to be $J\tran J$-conjugate, it holds that
	\begin{align} \label{eq:p_con}
		p_j\tran J\tran J p_{j+1} = 0.
	\end{align}
	Substituting (\ref{eq:directions}) in (\ref{eq:p_con}) gives 
	\begin{align}
		p_j\tran J\tran J p_{j+1} = p_j\tran J\tran J \hat{g}_{j+1} + \tau_j p_j\tran J\tran J p_j = 0,
	\end{align}
	from which it follows that
	\begin{align}
		\tau_j p_j\tran J\tran J p_j &=	- p_j\tran J\tran J \hat{g}_{j+1},  \\
		\tau_j &= - \frac{(J p_j)\tran (J \hat{g}_{j+1})}{(J p_j)\tran (J p_j)}. 
	\end{align}
\end{proof}

Note that although the system $J$ is unknown, the terms $J p_j$ \rev{and $\rev{J \hat{g}_{j+1}}$} in the expression for $\tau_j$ can be evaluated through experiments on the system.

\subsection{Step size selection} 

The optimal step size $\varepsilon_j$ in (\ref{eq:fjplus1}) for a general search direction $p_j$ is given by
\begin{align} 
	\varepsilon_j &= \arg \min_{\varepsilon} \mathcal{J}(f_{j+1}) \nonumber \\ \label{eq:opt_e}
	&= \arg \min_{\varepsilon} \|r-J(f_j + \varepsilon_j p_j)\|^2_2.
\end{align}
Since the criterion (\ref{eq:cost}) is quadratic, (\ref{eq:opt_e}) can be solved analytically. However, for a search direction $p_j$ that is based on gradient estimate $\hat{g}_j$, the standard conjugate gradient expression for $\varepsilon_j$ cannot be applied. Instead, the optimal step size for stochastic conjugate gradient ILC is given in the following theorem.

\begin{thm} \label{thm:stepsize}
	The optimal step size $\varepsilon_j$ that minimizes (\ref{eq:opt_e}) is given by
	\begin{align} \label{eq:stepsize}
		\varepsilon_j =  \rev{\frac{ e_j \tran (J p_j)}{(J p_j)\tran (J p_j)}}.
	\end{align}
\end{thm}

\begin{proof}
	Criterion (\ref{eq:opt_e}) is minimized by setting
	\begin{align}
		\frac{\partial \mathcal{J}(f_{j+1})}{\partial \varepsilon} = 0.
	\end{align}
	It holds that
	\begin{align}
	&\frac{\partial}{\partial \varepsilon}	\left(r-J(f_j + \varepsilon p_j) \right)\tran \left(r-J(f_j + \varepsilon p_j) \right) \\ \nonumber
	&= f_j\tran J\tran J p_j + p_j\tran J\tran J f_j + \rev{2} \varepsilon p_j \tran J\tran J p_j - p_j \tran J\tran r - r\tran J p_j,
	\end{align}
	from which it follows that
	\begin{align}
		\varepsilon_j&= \frac{p_j\tran J\tran (r-Jf_j) + (r-J f_j)\tran J p_j}{\rev{2}p_j\tran J\tran J p_j} \nonumber \\
		&= \rev{\frac{e_j \tran (J p_j)}{(J p_j)\tran (J p_j)}}.
	\end{align}
\end{proof}

\subsection{Unbiased gradient estimates through experiments on $J$}
\label{subsec:lcss}

Unbiased estimates of the gradient 
\begin{align} \label{eq:grad2}
		\rev{g(f_j) = 2 J\tran J f_j- 2 J\tran r = -2 J\tran e_j}
\end{align}
can be generated through experiment on $J$ by noting that $J\tran$ is the adjoint operator of $J$ and relates to $J$ through a time reversal, as described in the following.

\begin{defn}
	Let $\langle f,g \rangle = f\tran g$ denote the inner product of two signals $f,g \in \mathbb{R}^{N \times 1}$. The adjoint $J^*$ of $J$ is defined as the operator that satisfies the condition
	\begin{align} \nonumber
		\langle f,Jg \rangle = \langle J^* f, g \rangle, \quad \forall f,g \in \mathbb{R}^{N\times 1}.
	\end{align}
	The adjoint $J^*$ of $J$ is given by $J\tran$, which follows from
	\begin{align} \nonumber
		f^\top Jg = (J^* f)\tran g = f\tran (J^*)\tran g, \quad \forall f,g \in \mathbb{R}^{N \times 1}.
	\end{align}
\end{defn}

\begin{lem} \label{lem:adj_siso}
	The adjoint of a SISO system $J = J^{11}$ is given by $\left(J^{11}\right)\tran = \mathcal{T} J^{11} \mathcal{T}$, where the involutory permutation matrix 
	\begin{align*} 
		\mathcal{T} = \begin{bmatrix}
			0 & \dots &  0 & 1 \\
			\vdots &  &  1 & 0 \\
			0 & \iddots &  & \vdots \\
			1 & 0 & \dots & 0  
		\end{bmatrix} \in \mathbb{R}^{N \times N}
	\end{align*}
	has the interpretation of a time-reversal operator. 
\end{lem}

For SISO systems, Lemma \ref{lem:adj_siso} enables the measurement of the gradient using a single experiment by performing two time reversals. However, this is not applicable to non-symmetric MIMO systems, as is shown next.

\begin{lem} \label{lem:adj_mimo}
	The adjoint of a MIMO system $J$ is given by
	\begin{align} \label{eq:adj_mimo} 
		J\tran &=  \begin{bmatrix} (J^{11})\tran & \dots & (J^{n_o 1})\tran \\ \vdots & & \vdots \\ (J^{1 n_i})\tran & \dots & (J^{n_o n_i})\tran \end{bmatrix}   \\ \nonumber
		&=\underbrace{\begin{bmatrix} \mathcal{T} & & 0 \\ & \ddots & \\ 0 & & \mathcal{T} \end{bmatrix}}_{\mathcal{T}^{n_i}} 
		\underbrace{\begin{bmatrix} J^{11} & \dots & J^{n_o 1} \\ \vdots & & \vdots \\ J^{1 n_i} & \dots & J^{n_o n_i} \end{bmatrix}}_{\tilde{J}} \underbrace{\begin{bmatrix} \mathcal{T} & & 0 \\ & \ddots & \\ 0 & & \mathcal{T} \end{bmatrix}}_{\mathcal{T}^{n_o}}.
	\end{align}
\end{lem}

Since $\tilde{J} \neq J$ for non-symmetric MIMO systems, the term $J\tran e_j$, and therefore the gradient (\ref{eq:grad2}), cannot be determined exactly from a single experiment on $J$. However, it is possible to obtain an unbiased estimate $\hat{g}(f_j)$ of the gradient, for which $\mathbb{E}(\hat{g}(f_j)) = g(f_j)$, from a single experiment on $J$ as described in the following \cite{Aarnoudse2020c}.

\begin{lem}
	An unbiased estimate $\hat{g}(f)$ of (\ref{eq:grad}) is given by
	\begin{align} \label{eq:grad_bar}
		\hat{g}(f_j) = -2 \mathcal{T}^{n_i} A_j J A_j \mathcal{T}^{n_o} e_j.
	\end{align}
	The matrix $A_j \in \mathbb{R}^{(N n_i) \times (N n_o)}$ is given by
	\begin{align} \label{eq:A} 
		A_j = \begin{bmatrix}
			a^{11}_j & \dots & a^{1n_o}_j\\
			\vdots & \ddots &  \vdots \\
			a^{n_i 1}_j & \dots & a^{n_i n_o}_j 
		\end{bmatrix} \otimes I_N
	\end{align}
	where $I_N$ is the $N\times N$ identity matrix and the entries $a^{lm}_j$ are samples from a symmetric Bernoulli $\pm 1$ distribution, i.e., $a^{lm}_j\in \{-1,1\}$ and the probabilities are given by $P(a^{lm}_j = 1) = 1/2$ and $P(a^{lm}_j = -1) = 1/2$. 
\end{lem}

The unbiased gradient estimates generated according to (\ref{eq:grad_bar}) can be used in the stochastic gradient descent ILC algorithm (\ref{eq:fjplus1}), \rev{with conjugate directions according to (\ref{eq:directions}) and Theorem \ref{thm:directions}, and optimal step sizes given by (\ref{eq:stepsize}).}

%
%
%

\subsection{Implementation of conjugate gradient ILC}

The implementation of conjugate gradient ILC is outlined in Algorithm \ref{alg:cg}. \rev{Note that} $J p_{j-1}$ does not need to be measured in order to determine $\tau_{j-1}$, as it is already known from the computation of $\varepsilon_{j-1}$. As such, determining the conjugate search direction and the step size each require one dedicated experiment per iteration.

\begin{algorithm}[t]
	\caption{Stochastic Conjugate Gradient ILC} 	\label{alg:cg}
	\begin{algorithmic}[1]
		\State{\textbf{for} $j=1:n_{\text{iteration}}$}
		\State{\quad Apply input $f_j$ and measure $e_j = r - J f_j$.}
		\State{\quad Find approximation $\hat{g}_j$ of $g_j = -2J\tran e_j$ using }
		\Statex{\quad one experiment according to (\ref{eq:grad_bar}).}
		\State{\quad \textbf{if} $j=1$}
		\State{\qquad Set $p_1 = \hat{g}(f_1)$.}
		\State{\quad \textbf{else} } 		
		\State{\qquad Measure $J \hat{g}(f_j)$ and use $J p_{j-1}$ to find $\tau_{j-1}$ in (\ref{eq:tau}).}
		\State{\qquad Take direction $p_j = \hat{g}_j + \tau_{j-1} p_{j-1}$.}
		\State{\quad \textbf{end}}
		\State{\quad Measure $J p_{j}$ to find step size $\varepsilon_j$ in (\ref{eq:stepsize}).}	
		\State{\quad Update $f_{j+1} = f_j + \varepsilon_j p_j$.}		
		\State{\textbf{end}}	
	\end{algorithmic}
\end{algorithm}

\section{Analysis of related approaches} \label{sec:comp}

In this section, related methods are analyzed in comparison to the proposed stochastic conjugate gradient ILC method. First, the conjugate gradient method with exact gradients is analyzed. Secondly, search direction resets are proposed for situations where the available evaluations of the system $J$ are noisy. \rev{Lastly}, the gradient descent algorithm is recovered. 

\subsection{Conjugate gradient descent with exact gradients}
\label{subsection:exact_grad}

In the model-free conjugate gradient ILC approach proposed in the previous section, only approximate gradients are available. In the theoretical case with exact gradients and function evaluations, the expressions for the search direction and the step size can be simplified, resulting in the well-known standard expressions for the conjugate gradient method. If $\hat{g}_{j} = g_{j} \: \forall \: j$, expression (\ref{eq:tau}) for $\tau_j$ reduces to
\begin{align}
	\tau_j = \frac{g_{j+1}\tran g_{j+1}}{{g}_{j}\tran {g}_{j}},
\end{align}
see e.g.  \cite[Section 11.3]{Golub2013} for a derivation. In addition, expression (\ref{eq:stepsize}) for $\varepsilon_j$ reduces to
\begin{align} \label{eq:step_ideal}
	\varepsilon_{j} = - \frac{ g_j\tran  g_j}{(J p_j)\tran (J p_j)}.
\end{align}
Since these expressions are based on the assumption that exact gradients are known, they do not hold for  $\hat{g}_{j} \neq g_{j}$. As a result, the standard expressions for conjugate gradient descent cannot be used in the model-free conjugate gradient descent ILC algorithm.

\subsection{Stochastic conjugate gradient ILC with noisy system evaluations}
\label{subsec:noisy}
In typical control applications, the available evaluations of the system $J$ are noisy. While the influence of added noise on the stochastic gradients used in stochastic conjugate gradient ILC is limited, the search direction and step size that were previously assumed to be exact depend on evaluations of $J$ and are therefore influenced by the noise. As a result, in a noisy experimental setting it is not possible to maintain the conjugacy of search directions over multiple iterations. Therefore, it can be useful to reset the search direction to the gradient after a number of iterations, which is a common practice for Krylov subspace methods, see e.g. \cite{Baker2009}. In Section \ref{sec:sims}, the proposed approach is simulated for a situation with noisy system evaluations.

\subsection{Recovering gradient descent}

A gradient descent ILC algorithm is recovered from the conjugate gradient ILC algorithm by taking $\tau_j=0$, such that $p_j = \hat{g}_j$. Then, the input $f_{j+1}$ at iteration $j+1$ is given by
\begin{align} \label{eq:adjilc}
	f_{j+1} = f_j - \varepsilon_j \hat{g}(f_j),
\end{align} 
and the optimal step size is found to be
	\begin{align} 
	\varepsilon_j =  \frac{ e_j \tran (J \hat{g}(f_j))}{(J \hat{g}(f_j))\tran (J \hat{g}(f_j))}.
\end{align}
For exact conjugate gradient (CG) and gradient descent (GD), it holds that a CG step is at least as good as a GD step, provided that $J\tran J$ is symmetric positive definite \cite[Thm. 11.3.3]{Golub2013}, \cite{Allwright1976}. In \cite{Aarnoudse2020c} the gradient estimates $\hat{g}(f_j)$ are used in a Robbins-Monro type stochastic gradient descent algorithm \cite{Robbins1951} \rev{and} $\varepsilon_j$ is chosen such that criteria of almost sure convergence for a Robbins-Monro algorithm are met. As shown in Section \ref{sec:sims}, this results in slow convergence compared to the proposed stochastic conjugate gradient approach.

%
%

\section{Example} \label{sec:sims}

In this section, model-free conjugate gradient ILC is illustrated on a random non-symmetric $21\times21$ MIMO system with 84 states, generated using the function $\mathtt{drss}$ in MATLAB. A Bode plot of a part of the system is shown in Figure \ref{fig:bode}. The disturbance $r$ consists of a step in all directions. The approach is illustrated for examples with exact and noisy evaluations of the system $J$.

\subsection{Stochastic conjugate gradient}

The proposed stochastic conjugate gradient descent is compared to stochastic approximation adjoint ILC \cite{Aarnoudse2020c}, to deterministic adjoint ILC \cite{Bolder2018} and to a deterministic conjugate gradient method. In \cite{Bolder2018}, the gradient of criterion (\ref{eq:cost}) is generated through $n_i \times n_o$ experiments, structured as
\begin{align} \label{eq:jtranspose}
	g(f_j) = \mathcal{T}^{n_i} \left( \sum_{l=1}^{n_i} \sum_{m=1}^{n_o} E^{lm} J E^{lm} \right) \mathcal{T}^{n_o} e_j,%
\end{align}%
where $E^{lm}$ consists of zeros, with a one on the $lm^{\mathrm{th}}$ entry. The deterministic conjugate gradient method from Section \ref{subsection:exact_grad} is implemented using deterministic gradients generated by (\ref{eq:jtranspose}).%

\begin{figure}
	\setlength\figureheight{.4\figurewidth}%
	\begin{subfigure}{.9\linewidth}%
		\centering
		\includegraphics{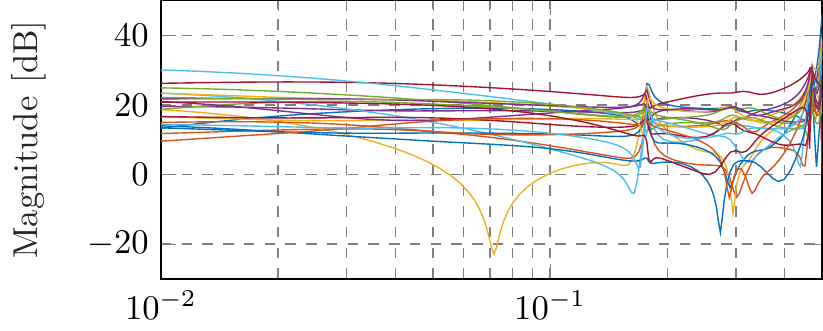}
	\end{subfigure}\\%
	\begin{subfigure}{.9\linewidth}%
		\centering
		\includegraphics{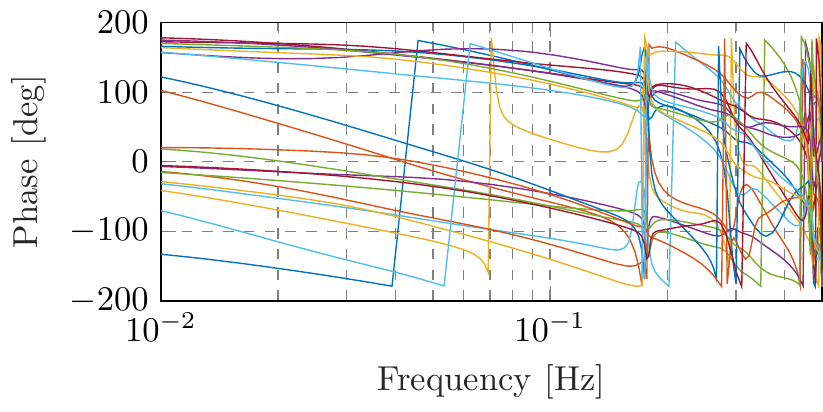}
	\end{subfigure}%
	\caption{Bode magnitude and phase plots of the last column of 21 subsystems of the random non-symmetric $21\times21$ MIMO system used in the example. \label{fig:bode}}%
\end{figure}%


\setlength\figureheight{.4\figurewidth}

\begin{figure}
	\centering
	\includegraphics{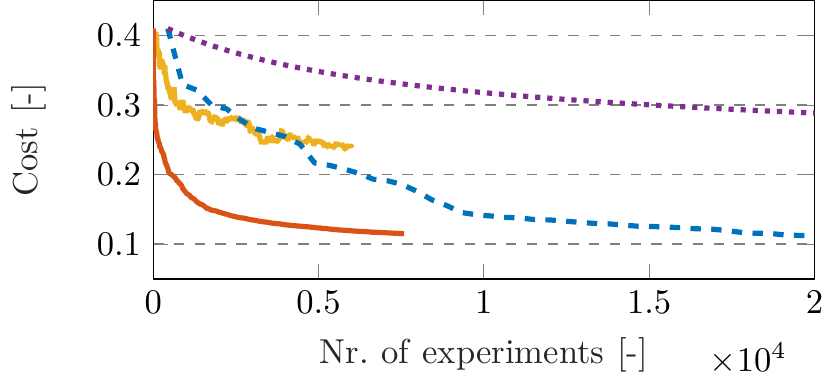}
	\caption{The cost as a function of the number of experiments for a non-symmetric $21\times 21$ system. Four approaches are shown: the proposed stochastic conjugate gradient ILC method (\protect\redline), deterministic conjugate gradient ILC (\protect\bluedash), stochastic gradient descent ILC (\protect\yelline) and deterministic gradient descent ILC (\protect\purdots). Stochastic conjugate gradient descent ILC requires far fewer experiments than other approaches to reach the same cost.}
	\label{fig:com_nonoise}
\end{figure}

In Figure \ref{fig:com_nonoise}, it is shown that the proposed stochastic conjugate gradient algorithm requires far fewer experiments to reach the same cost than the deterministic conjugate gradient algorithm. Both the stochastic and the deterministic conjugate gradient algorithm outperform the gradient descent algorithms. In addition, the smoothness of stochastic ILC is greatly improved by the line searches used in the conjugate gradient method.

\subsection{Stochastic conjugate gradients for noisy system evaluations} 

In typical control applications, evaluations of the system $J$ are noisy. As a result, gradient estimates generated through (\ref{eq:jtranspose}) are not deterministic, although the variance of these gradient estimates is typically smaller than that of those generated by (\ref{eq:grad_bar}). \rev{For a situation with noisy system evaluations, the proposed stochastic conjugate gradient ILC approach is compared to the conjugate gradient method of Section \ref{subsection:exact_grad}, which is designed for deterministic gradients.}%
	
\rev{In Figure \ref{fig:com_noise}, it is shown that the deterministic method diverges when noisy function evaluations are used. The stochastic conjugate gradient ILC algorithm is implemented with gradient estimates generated by respectively (\ref{eq:grad_bar}) and (\ref{eq:jtranspose}). It is shown that while the proposed algorithm converges for both gradient estimates, using gradient estimates obtained from a single experiment results in much faster convergence in terms of the required number of experiments.}


\begin{figure}[t]	
	\centering
	\includegraphics{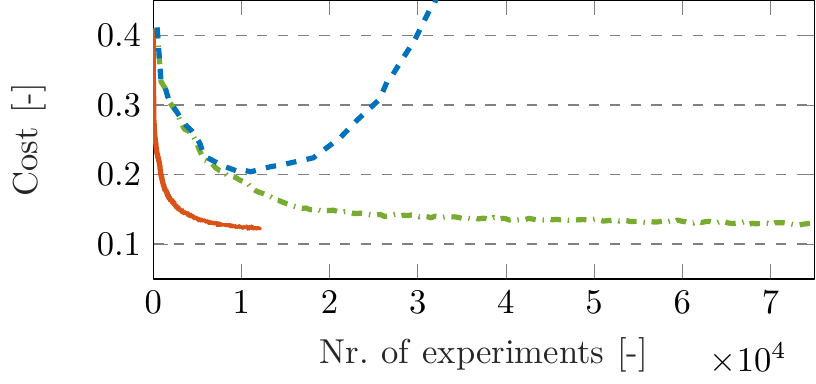}
	\caption{The cost as a function of the number of experiments for a noisy non-symmetric $21\times 21$ system. Stochastic conjugate gradient descent ILC based on (\ref{eq:grad_bar}) (\protect\redline) requires far fewer experiments than the implementation using (\ref{eq:jtranspose}) (\protect\greendashdot). Implementing the deterministic conjugate gradient (\protect\bluedash) method in case of noisy system evaluations and gradient estimates leads to divergent behavior.}
	\label{fig:com_noise}
\end{figure}

\section{Conclusions} \label{sec:conclusions}

A data-driven conjugate gradient method for model-free iterative learning control is introduced that uses efficient unbiased gradient estimates, such that it is suitable for massive MIMO systems. Dedicated experiments are used to determine conjugate search directions based on stochastic gradient estimates and corresponding optimal step sizes, resulting in a stochastic conjugate gradient ILC method. Compared to a deterministic model-free conjugate gradient method and to deterministic and stochastic model-free gradient descent algorithms, the proposed stochastic conjugate gradient method requires far fewer experiments to reach the same cost. In addition, it is shown that the proposed method converges when system evaluations are noisy, as opposed to the algorithm designed for deterministic gradients. \rev{Future developments involve extension of the framework to ILC with basis functions, experimental implementation and embedding in related iterative frameworks.}


\addtolength{\textheight}{-12cm}   



\bibliography{IEEEabrv,library}

\end{document}